\documentclass[a4paper]{article} 

\usepackage{microtype}
\usepackage[utf8]{inputenc}
\usepackage{lmodern}
\usepackage[T1]{fontenc}
\usepackage[mathscr]{eucal}
\usepackage[tbtags,fleqn]{amsmath}
\usepackage{amssymb}
\usepackage{amsthm}
\usepackage{authblk}
\usepackage[USenglish]{babel}
\usepackage[a4paper]{geometry}

\usepackage{xspace}
\usepackage{stackengine}

\theoremstyle{plain}
\newtheorem{theorem}{Theorem}
\newtheorem{lemma}[theorem]{Lemma}
\newtheorem{corollary}[theorem]{Corollary}
\newtheorem{proposition}[theorem]{Proposition}
\theoremstyle{definition}
\newtheorem{definition}[theorem]{Definition}
\newtheorem{example}[theorem]{Example}
\theoremstyle{remark}
\newtheorem*{remark}{Remark}

\newcommand{\MLm}{\mathsf{MATLANG}}
\newcommand{\ML}{$\MLm$\xspace}
\newcommand{\ARAm}{\mathsf{ARA}}
\newcommand{\ARA}{$\ARAm$\xspace}
\newcommand{\ARAC}{$(\ARAm+\zeta_k)(k)$\xspace}
\newcommand{\ARACTWO}{$(\ARAm+\zeta_2)(2)$\xspace}
\newcommand{\Rel}{\mathrm{Rel}}
\newcommand{\Mat}{\mathrm{Mat}}

\DeclareMathOperator{\sdiff}{\triangle}

\title{On matrices and $K$-relations}

\author{Robert Brijder}
\author{Marc Gyssens}
\author{Jan Van den Bussche}
\affil{Hasselt University, Belgium}
\date{}

\usepackage{color}

\begin{document}

\maketitle

\begin{abstract}
We show that the matrix query language \ML corresponds to a
natural fragment of the positive relational algebra on
$K$-relations.  The fragment is defined by introducing a
composition operator and restricting $K$-relation arities to two.
We then proceed to show that \ML can express all matrix
queries expressible in the positive relational algebra on
$K$-relations, when intermediate arities are restricted to three.
Thus we offer an analogue, in a model with numerical data, to the
situation in classical logic, where the algebra of binary
relations is equivalent to first-order logic with three
variables.
\end{abstract}

\section{Introduction}

Motivated by large-scale data science, there is recent interest
in supporting linear algebra operations, such as matrix
multiplication, in database systems.  This has prompted
investigations comparing the expressive power of common matrix
operations with the operations on relations provided by the
relational algebra and SQL
\cite{hutchison,polystorelalg,jermaine_linalgebra,matlang_icdt}.

For \emph{boolean} matrices, the connection between matrices and
relations is very natural and well known.  An $m \times n$
boolean matrix $A$ can be viewed as a binary relation $R
\subseteq \{1,\dots,m\} \times \{1,\dots,n\}$, where $R$ consists
of those pairs $(i,j)$ for which $A_{i,j} = 1$.  Boolean matrix
multiplication then amounts to composition of binary relations.
Composition is the central operation in the \emph{algebra of
binary relations}
\cite{tarski_relcalc,maddux_originra,pratt_relcalc}.  Besides
composition, this algebra has operations such as converse, which
corresponds to transposition of a boolean matrix; union and
complement, which correspond to disjunction and negation of
boolean matrices; and the empty and identity relations, which
correspond to the zero and identity matrices.

A common theme in research in the foundations of databases is the
expressive power of query languages \cite{ahv_book}.  When we
employ a query language, we would like to understand as well as
possible what we can do with it.  Of this kind is the classical
Codd theorem, stating the equivalence between the standard
relational algebra and first-order logic.  Likewise, for the
algebra of binary relations, a classical result
\cite{tarskigivant} is that it has the same expressive power as
the formulas with two free variables in $\mathrm{FO}(3)$, the three-variable
fragment of first-order logic.  In this sense, we understand
quite well the expressive power of a natural set of operations on
boolean matrices.

What can now be said in this regard
about more general matrices, with entries that are not just
boolean values?  An $m \times n$ matrix with entries in some
semiring $K$ is essentially a mapping from $\{1,\dots,m\} \times
\{1,\dots,n\}$ to $K$.  This perfectly fits the data model of
\emph{$K$-relations} introduced by Green, Garvounarakis and
Tannen \cite{provsemirings}.  In general, consider an infinite
domain $\mathbf{dom}$ and a supply of attributes.  In a database instance,
we assign to each attribute a range of values, in the form of a
finite subset of $\mathbf{dom}$.  Attributes can be declared to be
compatible; compatible attributes have the same range.  A
relation schema $S$ is a finite set of attributes.  Tuples over
$S$ are mappings that assign to each attribute a value of the
appropriate range.  Now a $K$-relation over $S$ is a mapping that
assigns to each tuple over $S$ an element of $K$.

So, an $m \times n$ matrix $X$ can be seen as a $K$-relation over
two attributes $A$ and $B$ where the range of $A$ is
$\{1,\dots,m\}$ and the range of $B$ is $\{1,\dots,n\}$.  We can
assume an order on all attributes and choose $A<B$ so that we
know which values are row indices and which are column indices.
Then an $n \times k$ matrix $Y$ is modeled using attributes $C<D$
where we choose $C$ and $B$ compatible, to reflect that the
number of columns of matrix $X$ equals the number of rows of
matrix $Y$.  We can view vectors as $K$-relations over a single
attribute, and scalars as $K$-relations over the empty schema.
In general, a $K$-relation of arity $r$ is essentially an
$r$-dimensional tensor (multidimensional array).
(Because we need not necessarily assume an order on
$\mathbf{dom}$, the tensor is unordered.)

Green et al.\ defined a generalization of the positive relation
algebra working on $K$-relations, which we denote here by
\ARA.\footnote{\ARA stands for Annotated-Relation Algebra, as the
elements from $K$ that a $K$-relation assigns to its tuples
are usually viewed as annotations.}  When we restrict
\ARA to arities of at most three, which we denote by $\ARAm(3)$, we
obtain an analogue to $\mathrm{FO}(3)$ mentioned above.  So, \ARA provides a
suitable scenario to reinvestigate, in a data model with
numerical values, the equivalence between the algebra of binary
relations and $\mathrm{FO}(3)$.  In this paper we offer the following
contributions.

\begin{enumerate}

\item

We define a suitable generalization, to $K$-relations, of the
composition operation of classical binary relations.  When we add
this composition operator to \ARA, but restrict arities to at most
two, we obtain a natural query language for matrices.  We refer
    to this language here as ``$\ARAm(2)$ plus composition''.

\item

  We show that $\ARAm(2)$ plus composition actually coincides with the
matrix query language \ML, introduced by two of the present
    authors with
Geerts and Weerwag \cite{matlang_icdt} in an attempt to
formalize the set of common matrix operations found in numerical
software packages.

\item

  We show that a matrix query is expressible in $\ARAm(3)$ if and only
    if it is expressible in \ML, thus providing an analogue
    to the classical result about $\mathrm{FO}(3)$ and the algebra of binary
    relations.  More generally, for any arity $r$, we show that
    an $r$-ary query over $r$-ary $K$-relations is expressible in
    $\ARAm(r+1)$ if and only if it is expressible in $\ARAm(r)$ plus
    composition.  For this result, we need the assumption that
    $K$ is commutative.  We stress that
    the proof is not a trivial adaptation of
    the proof of the classical result, because we can no longer
    rely on familiar classical properties like idempotence of
    union and join.

\end{enumerate}

\ARA has been a very influential vehicle for data
provenance.\footnote{The paper \cite{provsemirings} received the
PODS 2017 test-of-time award.} The elements from $K$ are
typically viewed as annotations, or even as abstract tokens, and
the semantics of \ARA operations was originally designed to show
how these annotations are propagated in the results of data
manipulations.  Other applications, apart from provenance, have
been identified from the outset, such as security levels, or
probabilities \cite{provsemirings}.  By doing the present work,
we have understood that \ARA can moreover serve as a
fully-fledged query language for tensors (multidimensional
arrays), and matrices in particular.  This viewpoint is
backed by the recent interest
in processing Functional Aggregate
Queries (FAQ \cite{faq,faq_sigmodrecord}, also known as AJAR
\cite{ajar}).  Indeed, FAQ and AJAR correspond to the
project-join fragment of \ARA, without self-joins.

This paper is further organized as follows.
Section~\ref{secprelim} recalls the
data model of $K$-relations and the associated query language
\ARA{}\@.  Section~\ref{secresult} presents the result on $\ARAm(r+1)$
and $\ARAm(r)$ plus composition.  Section~\ref{sec:matrices}
relates $\ARAm(2)$ plus composition to \ML{}\@.
Section~\ref{seconcl} draws conclusions, discusses related work,
and proposes directions for further research.

\section{Annotated-Relation Algebra} \label{secprelim}

By \emph{function} we will always mean a total function. For a function $f: X \to Y$ and $Z \subseteq X$, the \emph{restriction} of $f$ to $Z$, denoted by $f|_Z$, is the function $Z \to Y$ where $f|_Z(x) = f(x)$ for all $x \in Z$.

Recall that a \emph{semiring} $K$ is a set equipped with two binary operations, addition ($+$) and multiplication ($*$), such that (1) addition is associative, commutative, and has an identity element $0$; (2) multiplication is associative, has an identity element $1$; and has $0$ as an annihilating element; and (3) multiplication distributes over addition. A semiring is called \emph{commutative} when multiplication is commutative. We fix a semiring $K$.

\begin{remark}
We will explicitly indicate where we assume commutativity of $K$.
\end{remark}

From the outset, we also fix countable infinite sets $\mathbf{rel}$, $\mathbf{att}$, and $\mathbf{dom}$, the elements of which are called \emph{relation names}, \emph{attributes}, and \emph{domain elements}, respectively. We assume an equivalence relation $\sim$ on $\mathbf{att}$ that partitions $\mathbf{att}$ into an infinite number of equivalence classes that are each infinite. When $A \sim B$, we say that $A$ and $B$ are \emph{compatible}. Intuitively, $A \sim B$ will mean that $A$ and $B$ have the same set of domain values. A function $f: X \to Y$ with $X$ and $Y$ sets of attributes is called \emph{compatible} if $f(A) \sim A$ for all $A \in X$.

A \emph{relation schema} is a finite subset of $\mathbf{att}$. A \emph{database schema} is a function $\mathcal{S}$ on a finite set $N$ of relation names, assigning a relation schema $\mathcal{S}(R)$ to each $R \in N$.
The \emph{arity} of a relation name $R$ is the cardinality $|\mathcal{S}(R)|$ of its schema. The \emph{arity} of $\mathcal{S}$ is the largest arity among relation names $R \in N$.

We now recursively define the expressions of the
\emph{Annotated-Relation Algebra}, abbreviated by \ARA.
At the same time we assign a relation schema to each \ARA expression by extending $\mathcal{S}$ from relation names to arbitrary \ARA expressions. An \emph{\ARA expression} $e$ over a database schema $\mathcal{S}$ is equal to 
\begin{itemize}
\item a relation name $R$ of $\mathcal{S}$;
\item $\mathbf{1}(e')$, where $e'$ is an \ARA expression, and $\mathcal{S}(e) := \mathcal{S}(e')$;
\item $e_1 \cup e_2$, where $e_1$ and $e_2$ are \ARA expressions with $\mathcal{S}(e_1) = \mathcal{S}(e_2)$, and $\mathcal{S}(e) := \mathcal{S}(e_1)$;
\item $\pi_Y(e')$, where $e'$ is an \ARA expression and $Y \subseteq \mathcal{S}(e')$, and $\mathcal{S}(e) := Y$;
\item $\sigma_{Y}(e')$, where $e'$ is an \ARA expression, $Y \subseteq \mathcal{S}(e')$, the elements of $Y$ are mutually compatible, and $\mathcal{S}(e) := \mathcal{S}(e')$;
\item $\rho_\varphi(e')$, where $e'$ is an \ARA expression and $\varphi: \mathcal{S}(e') \to Y$ a compatible one-to-one correspondence with $Y \subseteq \mathbf{att}$, and $\mathcal{S}(e) := Y$; or
\item $e_1 \Join e_2$, where $e_1$ and $e_2$ are \ARA expressions, and $\mathcal{S}(e) := \mathcal{S}(e_1) \cup \mathcal{S}(e_2)$.
\end{itemize}

A \emph{domain assignment} is a function $D: \mathbf{att} \to
\mathcal{D}$, where $\mathcal{D}$ is a set of nonempty
finite subsets of $\mathbf{dom}$, such that $A \sim B$ implies
$D(A) = D(B)$. Let $X$ be a relation schema. A \emph{tuple} over
$X$ with respect to $D$
is a function $t: X \to \mathbf{dom}$ such that
$t(A) \in D(A)$ for all $A \in X$. We denote by
$\mathcal{T}_D(X)$ the set of tuples over $X$ with respect to $D$. Note that
$\mathcal{T}_D(X)$ is finite.  A \emph{relation} $r$ over
$X$ with respect to $D$ is a function $r:
\mathcal{T}_D(X) \to K$. So a relation annotates every tuple
over $X$ with respect to $D$ with a value from $K$.  If $\mathcal{S}$ is a
database schema, then an \emph{instance $\mathcal{I}$ of
$\mathcal{S}$ with respect to $D$} is a function that assigns to every
relation name $R$ of $\mathcal{S}$ a relation $\mathcal{I}(R):
\mathcal{T}_D(\mathcal{S}(R)) \to K$.

\begin{remark}
In practice, a domain assignment need only be defined on the
attributes that are used in the database schema (and on
attributes compatible to these attributes).  Thus, it can be
finitely specified.  While here we have chosen to keep the notion of domain
assignment and instance separate, it may perhaps be more natural
to think of the domain assignment as being part of the instance.
\end{remark}

\begin{figure}
\begin{center}
$\mathcal{I}(\text{no\texttt{\_}courses}) = $
\begin{tabular}{ |c|c|c| } 
\hline
student & dptm & $K$ \\
\hline
Alice & CS & $5$ \\
Alice & Math & $2$ \\
Alice & Bio & $0$ \\
Bob & CS & $2$ \\
Bob & Math & $1$ \\
Bob & Bio & $3$ \\
\hline
\end{tabular}
\qquad \qquad
$\mathcal{I}(\text{course\texttt{\_}fee}) = $
\begin{tabular}{ |c|c| } 
\hline
dptm & $K$ \\
\hline
CS & $300$ \\
Math & $250$ \\
Bio & $330$ \\
\hline
\end{tabular}
\end{center}
\caption{Example of a database instance.}
\label{fig:database_instance}
\end{figure}

\begin{example}\label{ex:db_instance}
Let us record for a university both the number of courses each student takes in each department and the course fee for each department. Let $K$ be the set of integers and let $\mathcal{S}$ be a
database schema on $\{\text{no\texttt{\_}courses},\allowbreak\text{course\texttt{\_}fee}\}$ with $\mathcal{S}(\text{no\texttt{\_}courses}) = \{\text{student},\text{dptm}\}$ and
$\mathcal{S}(\text{course\texttt{\_}fee}) = \{\text{dptm}\}$.
Let $D$ be a domain assignment with $D(\text{student}) = \{\text{Alice},\text{Bob}\}$ and $D(\text{dptm}) = \{\text{CS},\text{Math},\text{Bio}\}$. A database instance $\mathcal{I}$ of $\mathcal{S}$ with respect to $D$ is shown in Figure~\ref{fig:database_instance}.
\end{example}

We now define the relation $\mathbf{1}_X$, as well as
the generalizations of the classical
operations from the positive relational algebra to work on
$K$-relations.

\begin{description}

\item[One] For every relation schema $X$,
  we define $\mathbf{1}_X: \mathcal{T}_D(X) \to K$ where $\mathbf{1}_X(t) = 1$ for every $t \in \mathcal{T}_D(X)$.

\item[Union] Let $r_1, r_2: \mathcal{T}_D(X) \to K$. Define $r_1 \cup r_2: \mathcal{T}_D(X) \to K$ as $(r_1 \cup r_2)(t) = r_1(t) + r_2(t)$.

\item[Projection] Let $r: \mathcal{T}_D(X) \to K$ and $Y \subseteq X$. Define $\pi_{Y}(r): \mathcal{T}_D(Y) \to K$ as
\[
(\pi_{Y}(r))(t) = \sum_{\substack{t' \in \mathcal{T}_D(X),\\ t'|_{Y} = t}} \!\! r(t').
\]

\item[Selection] Let $r: \mathcal{T}_D(X) \to K$ and $Y \subseteq X$ where the elements of $Y$ are mutually compatible. Define $\sigma_{Y}(r): \mathcal{T}_D(X) \to K$ such that
\[
(\sigma_{Y}(r))(t) =
\begin{cases}
r(t) & \text{if } t(A)=t(B) \text{ for all } A, B \in Y;\cr
0    & \text{otherwise}.
\end{cases}
\]

\item[Renaming] Let $r: \mathcal{T}_D(X) \to K$ and $\varphi: X \to Y$ a compatible one-to-one correspondence. We define $\rho_\varphi(r): \mathcal{T}_D(Y) \to K$ as $\rho_\varphi(r)(t) = r(t \circ \varphi)$.

\item[Join] Let $r_1: \mathcal{T}_D(X_1) \to K$ and $r_2: \mathcal{T}_D(X_2) \to K$. Define $r_1 \Join r_2: \mathcal{T}_D(X_1 \cup X_2) \to K$ as $(r_1 \Join r_2)(t) = r_1(t|_{X_1})*r_2(t|_{X_2})$.
\end{description}

The above operations provide semantics for \ARA in a natural manner. Formally, let $\mathcal{S}$ be a database schema, let $e$ be an \ARA expression over $\mathcal{S}$, and let $\mathcal{I}$ be an instance of $\mathcal{S}$. The \emph{output} relation $e(\mathcal{I})$ of $e$ under $\mathcal{I}$ is defined as follows. If $e = R$ with $R$ a relation name of $\mathcal{S}$, then $e(\mathcal{I}) := \mathcal{I}(R)$. If $e = \mathbf{1}(e')$, then $e(\mathcal{I}) := \mathbf{1}_{\mathcal{S}(e')}$. If $e = e_1 \cup e_2$, then $e(\mathcal{I}) := e_1(\mathcal{I}) \cup e_2(\mathcal{I})$. If $e = \pi_{X}(e')$, then $e(\mathcal{I}) := \pi_{X}(e'(\mathcal{I}))$. If $e = \sigma_{Y}(e')$, then $e(\mathcal{I}) := \sigma_{Y}(e'(\mathcal{I}))$. If $e = \rho_\varphi(e')$, then $e(\mathcal{I}) := \rho_\varphi(e'(\mathcal{I}))$. Finally, if $e = e_1 \Join e_2$, then $e(\mathcal{I}) := e_1(\mathcal{I}) \Join e_2(\mathcal{I})$.

\begin{remark}
The language \ARA is a slight variation of the $K$-annotated relational algebra as originally defined by Green et al.~\cite{provsemirings} to better suit
  our purposes. 

First of all, the original definition does not have a domain
  assignment $D: \mathbf{att} \to \mathcal{D}$ but instead a
  single domain common to all attributes (and it therefore also
  does not have a compatibility relation $\sim$). As such, the
  original definition corresponds to the case where database
  schemas and \ARA expressions use only mutually compatible
  attributes. We need our more general setting when we compare
  \ARA to \ML in Section~\ref{sec:matrices}.

  Also, here, we focus on equality selections, while the original
  paper does not fix the allowed selection predicates.  Finally,
  the original definition assumes zero-relations $\mathbf{0}_X$,
  while we instead use one-relations $\mathbf{1}_X$.

\end{remark}

The following observations, to the effect that some (but not all) classical relational-algebra equivalences carry over to the $K$-annotated setting, were originally made by Green et al.
\begin{proposition}[Proposition~3.4 of \cite{provsemirings}]
  \label{prop:properties_ARA}
The following properties and equivalences hold, where, for each given equivalence, we assume that the left-hand side is well defined.
\begin{itemize}
\item Union is associative and commutative.
\item Join is associative and distributive over union, i.e., $(r_1 \cup r_2) \Join r_3 = (r_1 \Join r_3) \cup (r_2 \Join r_3)$.
\item Any two selections commute.
\item Projection and selection commute when projection retains the attributes on which selection takes place.
\item Projection distributes over union, i.e., $\pi_{Y}(r_1 \cup r_2) = \pi_{Y}(r_1) \cup \pi_{Y}(r_2)$.
\item Selection distributes over union, i.e., $\sigma_{Y}(r_1 \cup r_2) = \sigma_{Y}(r_1) \cup \sigma_{Y}(r_2)$.
\item We have $\sigma_{Y}(r_1) \Join r_2 = \sigma_{Y}(r_1 \Join r_2)$ and $r_1 \Join \sigma_{Y}(r_2) = \sigma_{Y}(r_1 \Join r_2)$.
\item If $K$ is commutative, then join is commutative.
\end{itemize}
\end{proposition}

Note that idempotence of union and of join, i.e., $r \Join r = r
\cup r = r$, which holds for the classical relational algebra,
does \emph{not} in general hold for \ARA.

We supplement Proposition~\ref{prop:properties_ARA} with the following easy-to-verify properties.
\begin{lemma}\label{lem:properties_ARA_extra}
Let $r_1: \mathcal{T}_D(X_1) \to K$ and $r_2: \mathcal{T}_D(X_2) \to K$. 
\begin{itemize}
\item 
If $X_1 \cap X_2 \subseteq X \subseteq X_1 \cup X_2$, then $\pi_X(r_1 \Join r_2) = \pi_{X \cap X_1}(r_1) \Join \pi_{X \cap X_2}(r_2)$. 

\item If $Y_1, Y_2 \subseteq X_1$ where $Y_1 \cap Y_2 \neq \emptyset$ and the attributes of $Y_1$ and of $Y_2$ are mutually compatible, then $\sigma_{Y_2}(\sigma_{Y_1}(r_1)) = \sigma_{Y_1 \cup Y_2}(r_1)$.

\item If $\varphi: X_1 \cup X_2 \to X$ is a compatible one-to-one correspondence, then $\rho_{\varphi}(r_1 \Join r_2) = \rho_{\varphi|_{X_1}}(r_1) \Join \rho_{\varphi|_{X_2}}(r_2)$. If moreover $X_1 = X_2$, then $\rho_{\varphi}(r_1 \cup r_2) = \rho_{\varphi}(r_1) \cup \rho_{\varphi}(r_2)$.

\item If $Y \subseteq X_1$ and $\varphi: X_1 \to X$ is a compatible one-to-one correspondence, then $\rho_{\varphi}(\sigma_{Y}(r_1)) = \sigma_{\varphi(Y)}(\rho_\varphi(r_1))$, where $\varphi(Y) = \{ \varphi(y) \mid y \in Y \}$.
\end{itemize}
\end{lemma}

We also use the operation of
projecting away an attribute, i.e., $\hat\pi_A(e) :=
\pi_{\mathcal{S}(e) \setminus \{A\}}(e)$ if $A \in
\mathcal{S}(e)$. Note that conversely, $\pi_X(e) = (\hat\pi_{A_m}
\cdots \hat\pi_{A_1})(e)$ where $X = \mathcal{S}(e) \setminus
\{A_1,\ldots,A_m\}$ and the $A_i$'s are mutually distinct.
Projecting away, allowing one to deal with one attribute
at a time, is sometimes notationally more convenient.

\section{Composition and Equivalence} \label{secresult}

In this section we define an operation called $k$-composition and show that augmenting \ARA by composition allows one to reduce the required arity of the relations that are computed in subexpressions.

\begin{definition}\label{def:composition}
Let $k$ be a nonnegative integer and let $l \in \{1, \ldots, k\}$. Let $r_i: \mathcal{T}_D(X_i) \to K$ for $i \in \{1, \ldots, l\}$, let $X = X_1 \cup \cdots \cup X_l$, and let $A \in X_1 \cap \cdots \cap X_l$.

Define the \emph{$k$-composition} $\zeta_{A,k}(r_1,\ldots,r_l): \mathcal{T}_D(X \setminus \{A\}) \to K$ as
\[
(\zeta_{A,k}(r_1,\ldots,r_l))(t) = (\hat\pi_{A}(r_1 \Join \cdots \Join r_l))(t)
\]
for all $t \in \mathcal{T}_D(X \setminus \{A\})$. 
\end{definition}

Note that $\zeta_{A,k}$ takes at most $k$ arguments.

We denote by $\ARAm+\zeta_k$ the language obtained by extending \ARA with $k$-composition. Consequently, if $e_1,\ldots,e_l$ are $\ARAm+\zeta_k$ expressions with $l \leq k$ and $A \in \mathcal{S}(e_1) \cap \cdots \cap \mathcal{S}(e_l)$, then $e = \zeta_{A,k}(e_1,\ldots,e_l)$ is an $\ARAm+\zeta_k$ expression. Also, we let $\mathcal{S}(e) := (\mathcal{S}(e_1) \cup \cdots \cup \mathcal{S}(e_l)) \setminus \{A\}$.

Let $k$ be a nonnegative integer. We denote by $\ARAm(k)$ the
fragment of \ARA in which the database schemas are restricted to
arity at most $k$ and the relation schema of each subexpression
is of cardinality at most $k$.  In particular, join $e_1 \Join e_2$ is
only allowed if $|\mathcal{S}(e_1 \Join e_2)| \leq k$. The
fragment \ARAC is defined similarly.

From Definition~\ref{def:composition} it is apparent that \ARAC
is subsumed by $\ARAm(k+1)$. One of our main results
(Corollary~\ref{cor:ARA_kp1+k})
provides the converse
inclusion, when the database schemas and outputs are restricted
to arity at most $k$. To this end, we establish a normal form for
\ARA expressions. First we prove the following technical but
important lemma.  

\begin{lemma}\label{lem:proj_sel}
Let $r_1, \ldots, r_n$ be relations with relation schemas $X_1,\ldots, X_n$, respectively, and with respect to a domain assignment $D$. Assume that $A, B \in X_1 \cup \cdots \cup X_n$ are distinct and compatible. Define, for $i \in \{1,\ldots,n\}$,
\[
r'_i :=
\begin{cases}
r_i & \text{if } A \notin X_i; \cr
\rho_{A \to B}(r_i) & \text{if } A \in X_i, B \notin X_i; \cr
\hat\pi_A(\sigma_{\{A,B\}}(r_i)) & \text{if } A,B \in X_i,
\end{cases}
\]
where $A \to B$ denotes the one-to-one correspondence from $X_i$ to $(X_i \setminus \{A\}) \cup \{B\}$ that assigns $A$ to $B$ and keeps the remaining attributes fixed. 

Then
\[
\hat\pi_A(\sigma_{\{A,B\}}(r_1 \Join \cdots \Join r_n)) = r'_1 \Join \cdots \Join r'_n.
\]
\end{lemma}
\begin{proof}
Let $X$ be a finite set of attributes with $A, B \in X$ distinct and compatible. Let $r: \mathcal{T}_D(X) \to K$ be a relation and $t \in \mathcal{T}_D(X \setminus \{A\})$.

We have 
\begin{align}
(\hat\pi_A(\sigma_{\{A,B\}}(r)))(t) \ = \!\!\!\!
\sum_{\substack{u \in \mathcal{T}_D(X),\\ u|_{X\setminus \{A\}} = t}} \!\!\!\! (\sigma_{A=B}(r))(u) \ 
= \!\!\!\! \sum_{\substack{u \in \mathcal{T}_D(X),\\ u|_{X\setminus \{A\}} = t,\\ u(A) = u(B)}} \!\!\!\! r(u) \ \,
= \ \, r(\tilde t) \label{form:proj_sel}
\end{align}
where $\tilde t \in \mathcal{T}_D(X)$ is $\tilde t = t \cup \{(A,t(B))\}$. Thus, $\tilde t$ is obtained from $t$ by adding attribute $A$ with value $t(B)$. Indeed, the last summation of (\ref{form:proj_sel}) is over a single tuple $u$, namely $u = \tilde t$.

In particular, applying (\ref{form:proj_sel}) to $r_1 \Join \cdots \Join r_n$, we obtain
\[
(\hat\pi_A(\sigma_{\{A,B\}}(r_1 \Join \cdots \Join r_n)))(t) =
(r_1 \Join \cdots \Join r_n)(\tilde t) = r_1(\tilde t|_{X_1}) * \cdots * r_n(\tilde t|_{X_n}).
\]
Denote the schemas of the relations $r'_1, \ldots, r'_n$ by $X'_1,\ldots, X'_n$, respectively. Let $i \in \{1,\ldots,n\}$. We distinguish three cases.
\begin{itemize}
\item If $A \notin X_i$, then $\tilde t|_{X_i} = t|_{X_i}$. Hence $r_i(\tilde t|_{X_i}) = r'_i(t|_{X'_i})$. 

\item If $A \in X_i$ and $B \notin X_i$, then $\tilde t|_{X_i} = t|_{(X_i \setminus \{A\}) \cup \{B\}} \circ \varphi = t|_{X'_i} \circ \varphi$ with $\varphi = A \to B$. Hence, $r_i(\tilde t|_{X_i}) = (\rho_{A \to B}(r_i))(t|_{X'_i}) = r'_i(t|_{X'_i})$.

\item If $A,B \in X_i$, then, by (\ref{form:proj_sel}) but now applied to $r_i$ and $t|_{X_i \setminus \{A\}}$, we have $r_i(\tilde t|_{X_i}) = (\hat\pi_A(\sigma_{\{A,B\}}(r_i)))(t|_{X_i \setminus \{A\}}) = r'_i(t|_{X'_i})$.
\end{itemize}
Therefore, in each case we obtain $r_i(\tilde t|_{X_i}) = r'_i(t|_{X'_i})$. Consequently,
\[
r_1(\tilde t|_{X_1}) * \cdots * r_n(\tilde t|_{X_n}) = r'_1(t|_{X'_1}) * \cdots * r'_n(t|_{X'_n}) = (r'_1 \Join \cdots \Join r'_n)(t).\qedhere
\]
\end{proof}

We use the following terminology.  Let $\mathcal{F}$ be any
family of expressions. A \emph{selection of
$\mathcal{F}$-expressions} is an expression of the form
$\sigma_{Y_n} \cdots \sigma_{Y_1}(f)$, where $f$ is an
$\mathcal{F}$-expression and $n \geq 0$. Note the slight abuse of
terminology as we allow multiple selection operations. Also, when
we say that $e$ is a \emph{union of $\mathcal{F}$-expressions} or
a \emph{join of $\mathcal{F}$-expressions}, we allow $e$ to be
just a single expression in $\mathcal{F}$ (so union and join may
be skipped). 

We are now ready to formulate and prove a main result of this paper. This result is inspired by a result in \cite[Theorem 3.4.5, Claim 2]{marxvenema_multi} which provides a proof of the classic equivalence of $\mathrm{FO}(3)$ and the algebra of binary relations.

Two \ARA expressions $e_1$ and $e_2$ over the same database
schema are called \emph{equivalent}, naturally, if they yield the
same output relation, for every domain assignment and every
database instance respecting that domain assignment.

\begin{theorem}\label{thm:ara_equiv_bara}
Let $\mathcal{S}$ be a database schema of arity at most $k$ and assume that $K$ is commutative. Every $\ARAm(k+1)$ expression over $\mathcal{S}$ is equivalent to a union of selections of joins of \ARAC expressions over $\mathcal{S}$.
\end{theorem}
\begin{proof}
For brevity, if an expression $e$ is a union of selections of joins of \ARAC expressions over $\mathcal{S}$, then we say that $e$ is in \emph{normal form}.

The proof is by induction on the structure of $e$. 

\emph{Relation names.} Since relation names $R$ of $\mathcal{S}$ are of arity at most $k$, we have that $R$ is an \ARAC expression over $\mathcal{S}$.

\emph{One.} If $e$ is equivalent to a union of selections of joins of \ARAC expressions $e_1, \ldots, e_n$ over $\mathcal{S}$, then $\mathbf{1}(e) \equiv \mathbf{1}(e_1 \Join \cdots \Join e_n) \equiv \mathbf{1}(e_1) \Join \cdots \Join \mathbf{1}(e_n)$ and so $\mathbf{1}(e)$ is also equivalent to an expression in normal form.

\emph{Union.} If both $e_1$ and $e_2$ are equivalent to expressions in normal form, then so is $e_1 \cup e_2$.

\emph{Join.} Since join distributes over union and since $\sigma_{Y}(e_1) \Join e_2$ and $e_1 \Join \sigma_{Y}(e_2)$ are equivalent to $\sigma_{Y}(e_1 \Join e_2)$ (Proposition~\ref{prop:properties_ARA}), we observe that if $e_1$ and $e_2$ are equivalent to expressions in normal form, then so is $e_1 \Join e_2$.

Let $e$ be equivalent to an expression in normal form.

\emph{Selection.} Since selection distributes over union (Proposition~\ref{prop:properties_ARA}), $\sigma_{Y}(e)$ is also equivalent to an expression in normal form.

\emph{Renaming.} Since renaming distributes over union and join and by the commutative property of renaming and selection (Lemma~\ref{lem:properties_ARA_extra}), $\rho_{\varphi}(e)$ is also equivalent to an expression in normal form.

\emph{Projection.} In this case we additionally require that $|\mathcal{S}(e)| \leq k+1$ (which holds when $e$ is an $\ARAm(k+1)$ expression). Since projection distributes over union, it suffices to assume that $e = \sigma_{Y_m}\cdots\sigma_{Y_1}(e')$, where $e'$ is a join of \ARAC expressions. By Proposition~\ref{prop:properties_ARA} and Lemma~\ref{lem:properties_ARA_extra} we may assume that the $Y_i$'s are mutually disjoint. We may also assume that the $Y_i$'s are all of cardinality at least $2$ (since $\sigma_{Y}$ on relations is the identity when $|Y| \leq 1$).
Let $A \in \mathcal{S}(e)$. We prove that $\hat\pi_{A}(e)$ is a selection of a join of \ARAC expressions. 

If $A$ does not belong to any of the $Y_i$'s, then $\hat\pi_{A}(e) = \sigma_{Y_m}\cdots\sigma_{Y_1}(\hat\pi_{A}(e'))$ by Proposition~\ref{prop:properties_ARA}. If $A \in Y_i$ for some $i$, then this $i$ is unique since the $Y_i$'s are mutually disjoint. Since any two selections commute, we may assume that $A \in Y_1$. As before, $\hat\pi_{A}(e) = \sigma_{Y_m}\cdots\sigma_{Y_2}(\hat\pi_{A}(\sigma_{Y_1}(e')))$. Since $Y_1$ is of cardinality at least $2$, there exists $B \in Y_1$ distinct from $A$. We have $\sigma_{Y_1}(e') \equiv \sigma_{Y_1 \setminus \{A\}}(\sigma_{\{A,B\}}(e'))$. Therefore $\hat\pi_{A}(\sigma_{Y_1}(e')) \equiv \sigma_{Y_1 \setminus \{A\}}(\hat\pi_{A}(\sigma_{\{A,B\}}(e')))$. By Lemma~\ref{lem:proj_sel}, we obtain that $\hat\pi_{A}(\sigma_{\{A,B\}}(e'))$ is equivalent to a join of \ARAC expressions. So, this case is settled.

It remains to show that $\hat\pi_{A}(e')$ is a join of \ARAC expressions. If $|\mathcal{S}(e')| \leq k$, then $e'$ itself is an \ARAC expression and so is $\hat\pi_{A}(e')$. So, assume that $|\mathcal{S}(e')| = k+1$.

Since join is commutative (because $K$ is) and associative, we can regard $e'$ as a join of a \emph{multiset} $F$ of \ARAC expressions. By Lemma~\ref{lem:properties_ARA_extra}, for expressions $e_1$ and $e_2$, if $A \notin \mathcal{S}(e_1)$, then $\hat\pi_{A}(e_1 \Join e_2) \equiv e_1 \Join \hat\pi_{A}(e_2)$. Therefore, we may assume that for every $f \in F$, we have $A \in \mathcal{S}(f)$. Hence, with $\mathcal{P}$ the set of all $k$-element subsets of $\mathcal{S}(e')$ containing $A$, there exists a function $p$ that assigns to each $f \in F$ a set $S \in \mathcal{P}$ with $\mathcal{S}(f) \subseteq S$. Let $\mathcal{R}$ be the range of $p$. Thus $|\mathcal{R}| \leq |\mathcal{P}| = k$. Let, for $S \in \mathcal{R}$, $e_S := \, \Join_{f \in F, p(f) = S} f$. Note that each $e_S$ is an \ARAC expression since $|\mathcal{S}(e_S)| \leq |S| = k$. We have $e' \equiv \,\, \Join_{S \in \mathcal{R}} e_S$ whence $\hat \pi_A(e') \equiv \hat \pi_A(\Join_{S \in \mathcal{R}} e_S)$, which coincides with $\zeta_{A,k}((e_S)_{S \in \mathcal{R}})$ by Definition~\ref{def:composition}. We thus obtain an \ARAC expression as desired.
\end{proof}

\begin{example}
Assume that $K$ is commutative and consider the $\ARAm(3)$ expression
\[
e = \pi_{\{B,C\}}(\sigma_{\{B,C\}}(R \Join R \Join S \Join T \Join \rho_\varphi(T)) \cup \sigma_{\{A,B\}}(R \Join S \Join T)),
\]
where $\mathcal{S}(R) = \{A,B\}$, $\mathcal{S}(S) = \{B,C\}$, $\mathcal{S}(T) = \{A,C\}$ ($A,B,C$ are mutually distinct), and $\varphi$ sends $A$ to $B$ and $C$ to itself. The proof of Theorem~\ref{thm:ara_equiv_bara} obtains an equivalent expression in normal form as follows.
\begin{align*}
e &= \hat\pi_{A}(\sigma_{\{B,C\}}(R \Join R \Join S \Join T \Join \rho_\varphi(T)) \cup \sigma_{\{A,B\}}(R \Join S \Join T)) \cr
&\equiv \hat\pi_{A}(\sigma_{\{B,C\}}(R \Join R \Join S \Join T \Join \rho_\varphi(T))) \cup \hat\pi_{A}(\sigma_{\{A,B\}}(R \Join S \Join T)) \cr
&\equiv \sigma_{\{B,C\}}(\hat\pi_{A}(R \Join R \Join S \Join T \Join \rho_\varphi(T))) \cup \hat\pi_{A}(\sigma_{\{A,B\}}(R \Join S \Join T)) \cr
&\equiv \sigma_{\{B,C\}}(S \Join \rho_\varphi(T) \Join \hat\pi_{A}(R \Join R \Join T)) \cup \hat\pi_{A}(\sigma_{\{A,B\}}(R \Join S \Join T)) \cr
&\equiv \sigma_{\{B,C\}}(S \Join \rho_\varphi(T) \Join \zeta_{A,2}(R \Join R, T)) \cup \hat\pi_{A}(\sigma_{\{A,B\}}(R \Join S \Join T)) \cr
&\equiv \sigma_{\{B,C\}}(S \Join \rho_\varphi(T) \Join \zeta_{A,2}(R \Join R, T)) \cup \bigl(\hat\pi_{A}(\sigma_{\{A,B\}}(R)) \Join S \Join \rho_{A \to B}(T)\bigr).
\end{align*}
The last expression is in the normal form since $S$, $\rho_\varphi(T)$, $\zeta_{A,2}(R \Join R, T)$, $\hat\pi_{A}(\sigma_{\{A,B\}}(R))$, and $\rho_{A \to B}(T)$ are all \ARACTWO expressions.
\end{example}

Note that we likely cannot omit the ``selections of'' in
the above theorem. For example, for $k=2$ consider
$\sigma_{\{A,C\}}(R \Join S)$ where $R$ and $S$ are relation
names with $\mathcal{S}(R) = \{A,B\}$ and $\mathcal{S}(S) =
\{B,C\}$.

\begin{remark}
Theorem~\ref{thm:ara_equiv_bara} still holds if the $\mathbf{1}$ operator is omitted from the definition of \ARA. Indeed, in the proof we can simply omit the case for the $\mathbf{1}$ operator, which is not used anywhere else.
\end{remark}

Since union, selection, and join do not decrease the number of
attributes of relations, we have the following corollary to
Theorem~\ref{thm:ara_equiv_bara}, which establishes the main
result announced in the Introduction.
\begin{corollary}\label{cor:ARA_kp1+k}
Let $\mathcal{S}$ be a database schema of arity at most $k$ and assume that $K$ commutative. Every $\ARAm(k+1)$ expression $e$ over $\mathcal{S}$ with $|\mathcal{S}(e)| \leq k$ is equivalent to an \ARAC expression over $\mathcal{S}$.
\end{corollary}

\begin{remark}
We remark that transforming an expression into the normal form of
Theorem~\ref{thm:ara_equiv_bara} may lead to an exponential
increase in expression length. The reason is that the proof uses
distributivity of join over union.
Indeed, each time we replace an expression of the form
$(e_1 \cup e_2) \Join e_3$ by $(e_1 \Join e_3) \cup (e_2 \Join
e_3)$ there is a duplication of $e_3$. The proof of the classic
translation of $\mathrm{FO}(3)$ to the algebra of binary
relations also induces an exponential increase of expression
length for similar reasons. A proof that this blowup is
unavoidable remains open, both for our result and for
the classical result (to the best of our knowledge).
\end{remark}

\section{Matrices} \label{sec:matrices}

In this section we show that \ARACTWO is equivalent to a natural
version of \ML \cite{matlang_icdt}. As a consequence of
Corollary~\ref{cor:ARA_kp1+k}, we then obtain that also
$\ARAm(3)$, with database schemas and output relations restricted
to arity at most $2$, is equivalent to \ML.  We begin by
recalling the definition of this language.

\subsection{\ML}
Let us fix the countable infinite sets $\mathbf{matvar}$ and $\mathbf{size}$, where the latter has a distinguished element $1 \in \mathbf{size}$. The elements of $\mathbf{matvar}$ are called \emph{matrix variables} and the elements of $\mathbf{size}$ are called \emph{size terms}. 

A \emph{matrix schema} is a function $\mathcal{S}: V \to \mathbf{size} \times \mathbf{size}$ with $V \subseteq \mathbf{matvar}$ both finite and nonempty. We write $(\alpha,\beta) \in \mathbf{size} \times \mathbf{size}$ also as $\alpha \times \beta$.

\ML expressions are recursively defined as follows. At the same time we assign a matrix schema to each \ML expression by extending $\mathcal{S}$ from matrix variables to arbitrary \ML expressions. 

A \emph{\ML expression} $e$ over a matrix schema $\mathcal{S}$ is equal to
\begin{description}
\item[Variable] a matrix variable $M$ of $\mathcal{S}$;

\item[Transpose] $(e')^T$, where $e'$ is a \ML expression, and $\mathcal{S}(e) := \beta \times \alpha$ if $\mathcal{S}(e') = \alpha \times \beta$;

\item[One-vector] $\mathbf{1}(e')$, where $e'$ is a \ML expression, and $\mathcal{S}(e) := \alpha \times 1$ if $\mathcal{S}(e') = \alpha \times \beta$;

\item[Diagonalization] $\mathrm{diag}(e')$, where $e'$ is a \ML expression with $\mathcal{S}(e') = \alpha \times 1$, and $\mathcal{S}(e) := \alpha \times \alpha$;

\item[Multiplication] $e_1 \cdot e_2$, where $e_1$ and $e_2$ are \ML expressions with $\mathcal{S}(e_1) = \alpha \times \beta$ and $\mathcal{S}(e_2) = \beta \times \gamma$, and $\mathcal{S}(e) := \alpha \times \gamma$;

\item[Addition] $e_1 + e_2$, where $e_1$ and $e_2$ are \ML expressions with $\mathcal{S}(e_1) = \mathcal{S}(e_2)$, and $\mathcal{S}(e) := \mathcal{S}(e_1)$; or

\item[Hadamard product] $e_1 \circ e_2$, where $e_1$ and $e_2$ are \ML expressions with $\mathcal{S}(e_1) = \mathcal{S}(e_2)$, and $\mathcal{S}(e) := \mathcal{S}(e_1)$.
\end{description}

A \emph{size assignment} is a function $\sigma$ that assigns to
each size term a strictly positive integer with $\sigma(1) = 1$.
Let $\mathcal{M}$ be the set of all matrices over $K$. We say
that $M \in \mathcal{M}$ \emph{conforms} to $\alpha \times \beta
\in \mathbf{size} \times \mathbf{size}$ by $\sigma$ if $M$ is a
$\sigma(\alpha) \times \sigma(\beta)$-matrix.

If $\mathcal{S}: V \to \mathbf{size} \times \mathbf{size}$ is a
matrix schema, then an \emph{instance of $\mathcal S$ with
respect to $\sigma$} is a function $\mathcal{I}: V \to
\mathcal{M}$ such that, for each $M \in V$, the matrix
$\mathcal{I}(M)$ conforms to $\mathcal{S}(M)$ by $\sigma$.

\begin{remark}
In practice, a size assignment need only be defined on the
size terms that are used in the schema.
Thus, it can be
finitely specified.  While here we have chosen to keep the notion
of size
assignment and instance separate, it may perhaps be more natural
to think of the size assignment as being part of the instance.
\end{remark}

\begin{figure}
$$
\mathcal{I}(\text{no\texttt{\_}courses}) =
\begin{pmatrix}
5 & 2 & 0 \cr
2 & 1 & 3
\end{pmatrix}
\qquad \qquad
\mathcal{I}(\text{course\texttt{\_}fee}) =
\begin{pmatrix}
300\cr
250\cr
330
\end{pmatrix}
$$
\caption{An example of an instance of a matrix schema.}
\label{fig:matrixdb_instance}
\end{figure}

\begin{example}\label{ex:matrixdb_instance}
This example is similar to Example~\ref{ex:db_instance}. Let $K$ be the set of integers and let $\mathcal{S}$ be a matrix schema on $\{\text{no\texttt{\_}courses},\allowbreak\text{course\texttt{\_}fee}\}$ with $\mathcal{S}(\text{no\texttt{\_}courses}) = \text{student} \times \text{dptm}$ and $\mathcal{S}(\text{course\texttt{\_}fee}) = \text{dptm} \times 1$. Let $\sigma$ be a size assignment with $\sigma(\text{student}) = 2$ and $\sigma(\text{dptm}) = 3$. An instance $\mathcal{I}$ of $\mathcal{S}$ with respect to $\sigma$ is shown in Figure~\ref{fig:matrixdb_instance}.
\end{example}

The semantics for \ML is given by the following matrix operations. Let $A$ be an $m \times n$-matrix over $K$. We define $\mathbf{1}(A)$ to be the $m \times 1$-matrix (i.e., column vector) with $\mathbf{1}(A)_{i,1} = 1$. If $n=1$ (i.e., $A$ is a column vector), then $\mathrm{diag}(A)$ is the $m \times m$-matrix with $\mathrm{diag}(A)_{i,j}$ equal to $A_{i,1}$ if $i = j$ and to $0$ otherwise. If $B$ is an $m \times n$-matrix, then $A \circ B$ denotes the Hadamard product of $A$ and $B$. In other words, $(A \circ B)_{i,j} = A_{i,j}*B_{i,j}$. Matrix addition and matrix multiplication are as usual denoted by $+$ and $\cdot$, respectively.

Formally, let $\mathcal{S}$ be a matrix schema, let $e$ be a \ML expression over $\mathcal{S}$, and let $\mathcal{I}$ be a matrix instance of $\mathcal{S}$. Then the \emph{output} matrix $e(\mathcal{I})$ of $e$ under $\mathcal{I}$ is defined in the obvious way, given the operations just defined. If $e = M$ with $M$ a matrix variable of $\mathcal{S}$, then $e(\mathcal{I})$ is naturally defined to be equal to $\mathcal{I}(M)$.

\begin{remark} \label{rempoint}
Matrix addition and the Hadamard product are the pointwise
applications of addition and product, respectively. The original
definition of \ML \cite{matlang_icdt} is more generally defined
in terms of an arbitrary set $\Omega$ of allowed pointwise
functions. So, \ML as defined above fixes $\Omega$ to
$\{+,\cdot\}$. This restriction was also considered by Geerts~\cite{floris_lagraphs} (who also allows multiplication by constant scalars, but this is not essential).

Also, the original definition of \ML fixes $K$ to the field of complex numbers and complex transpose is considered instead of (ordinary) transpose. Of course, transpose can be expressed using complex transpose and pointwise application of conjugation.
\end{remark}

\begin{table}
  \begin{center}
    \begin{tabular}{l|l|l}
      Mapping & $\MLm \to \ARAm$ & $\ARAm \to \MLm$\\
      \hline
      attributes $A$/size terms $\alpha$ & $\mathrm{row}_\alpha$, $\mathrm{col}_\alpha$ & $\Psi(A)$\\
      schemas $\mathcal{S}$ & $\Gamma(\mathcal{S})$ & $\Theta(\mathcal{S})$\\
      expressions $e$ & $\Upsilon(e)$ & $\Phi(e)$\\
      instances $I$, relations $r$/matrices $M$ & $\Rel_{\mathcal{S},\sigma}(I)$, $\Rel_{s,\sigma}(M)$ & $\Mat_D(I)$, $\Mat_D(r)$
    \end{tabular}
  \end{center}
  \caption{Symbol table for the simulations between \ML and \ARACTWO.}
  \label{tab:symbol_table}
\end{table}

In the following subsections we provide simulations between \ML
and \ARACTWO.  The notations for the different translations that
will be given are summarized in Table~\ref{tab:symbol_table}.

\subsection{Simulating \ML in \ARACTWO}

For notational convenience, instead of fixing a one-to-one correspondence between $\mathbf{rel}$ and $\mathbf{matvar}$, we assume that $\mathbf{rel} = \mathbf{matvar}$.

Let us now fix injective functions $\mathrm{row}: \mathbf{size} \setminus \{1\} \to \mathbf{att}$ and $\mathrm{col}: \mathbf{size} \setminus \{1\} \to \mathbf{att}$ such that (1) $\mathrm{row}(\alpha)$ and $\mathrm{col}(\alpha)$ are compatible for all $\alpha \in \mathbf{size} \setminus \{1\}$ and (2) the range of $\mathrm{row}$ is disjoint from the range of $\mathrm{col}$. To reduce clutter, we also write, for $\alpha \in \mathbf{size} \setminus \{1\}$, $\mathrm{row}(\alpha)$ as $\mathrm{row}_\alpha$ and $\mathrm{col}(\alpha)$ as $\mathrm{col}_\alpha$.

Let $s \in \mathbf{size} \times \mathbf{size}$. We associate to $s$ a relation schema $\Gamma(s)$ with $|\Gamma(s)| \leq 2$ as follows.
\[
\Gamma(s) :=
\begin{cases}
\{\mathrm{row}_\alpha, \mathrm{col}_\beta\} & \text{if } s = \alpha \times \beta; \cr
\{\mathrm{row}_\alpha\} & \text{if } s = \alpha \times 1; \cr
\{\mathrm{col}_\beta\} & \text{if } s = 1 \times \beta; \cr
\emptyset & \text{if } s = 1 \times 1,
\end{cases}
\]
where $\alpha \neq 1 \neq \beta$.

Let $\mathcal{S}$ be a matrix schema on a set of matrix variables $V$. We associate to $\mathcal{S}$ a database schema $\Gamma(\mathcal{S})$ on $V$ as follows. For $M \in V$, we set $(\Gamma(\mathcal{S}))(M) := \Gamma(\mathcal{S}(M))$.

Let $\sigma$ be a size assignment. We associate to $\sigma$ a domain assignment $D(\sigma)$ where, for $\alpha \in \mathbf{size}$, $(D(\sigma))(\mathrm{row}_\alpha) := (D(\sigma))(\mathrm{col}_\alpha) := \{1,\ldots,\sigma(\alpha)\}$.

Let $M \in \mathcal{M}$ conform to $s = \alpha \times \beta$ by $\sigma$.
We associate to
$M$ a relation $\Rel_{s,\sigma}(M):
\mathcal{T}_{D(\sigma)}(\Gamma(s)) \to K$ as follows. We have
$(\Rel_{s,\sigma}(M))(t) := M_{i,j}$, where (1) $i$ is equal to
$t(\mathrm{row}_\alpha)$ if $\alpha \neq 1$ and equal to $1$ if
$\alpha = 1$; and (2) $j$ is equal to $t(\mathrm{col}_\beta)$ if
$\beta \neq 1$ and equal to $1$ if $\beta = 1$.

Let $\mathcal{S}: V \to \mathbf{size} \times \mathbf{size}$ be a
matrix schema and let $\mathcal{I}$ be a matrix instance of
$\mathcal{S}$ with respect to $\sigma$. We associate to
$\mathcal{I}$ an instance
$\Rel_{\mathcal{S},\sigma}(\mathcal{I})$ of database schema
$\Gamma(\mathcal{S})$ with respect to $D(\sigma)$ as follows. For $M \in V$, we set $(\Rel_{\mathcal{S},\sigma}(\mathcal{I}))(M) := \Rel_{\mathcal{S}(M),\sigma}(\mathcal{I}(M))$.

\begin{figure}
\begin{center}
$\mathcal{I}(\text{no\texttt{\_}courses}) = $
\begin{tabular}{ |c|c|c| } 
\hline
$\rm row_{\rm student}$ & $\rm col_{\rm dptm}$ & $K$ \\
\hline
1 & 1 & $5$ \\
1 & 2 & $2$ \\
1 & 3 & $0$ \\
2 & 1 & $2$ \\
2 & 2 & $1$ \\
2 & 3 & $3$ \\
\hline
\end{tabular}
\qquad \qquad
$\mathcal{I}(\text{course\texttt{\_}fee}) = $
\begin{tabular}{ |c|c| } 
\hline
$\rm row_{\rm dptm}$ & $K$ \\
\hline
1 & $300$ \\
2 & $250$ \\
3 & $330$ \\
\hline
\end{tabular}
\end{center}
\caption{Matrix instance from Figure~\ref{fig:matrixdb_instance}
represented as a database instance.}
\label{mamafig}
\end{figure}

\begin{example}
Recall $\mathcal{I}$, $\mathcal{S}$, and $\sigma$ from
Example~\ref{ex:matrixdb_instance}. We have that
$(\Gamma(\mathcal{S}))(\text{no\texttt{\_}courses}) =
\{\mathrm{row}_\text{student},\mathrm{col}_\text{dptm}\}$ and
$(\Gamma(\mathcal{S}))(\text{course\texttt{\_}fee}) =
\{\mathrm{row}_\text{dptm}\}$. The database instance
$\Rel_{\mathcal{S},\sigma}(\mathcal{I})$ is shown in
Figure~\ref{mamafig}.
\end{example}

We now show that every \ML expression can be simulated by an \ARACTWO expression.
\begin{lemma}\label{lem:MLtoARA}
For each \ML expression $e$ over a matrix schema $\mathcal{S}$, there exists an \ARACTWO expression $\Upsilon(e)$ over database schema $\Gamma(\mathcal{S})$ such that (1) $\Gamma(\mathcal{S}(e)) = (\Gamma(\mathcal{S}))(\Upsilon(e))$ and (2) for all size assignments $\sigma$ and matrix instances $\mathcal{I}$ of $\mathcal{S}$ with respect to $\sigma$, we have $\Rel_{\mathcal{S}(e),\sigma}(e(\mathcal{I})) = (\Upsilon(e))(\Rel_{\mathcal{S},\sigma}(\mathcal{I}))$.
\end{lemma}
\begin{proof}
Assume that $\mathcal{S}(e) = \alpha \times \beta$, where $\alpha, \beta \in \mathbf{size}$. We define $\Upsilon(e)$ explicitly. 
\begin{itemize}
\item If $e = M$ is a matrix variable of $\mathcal{S}$, then $\Upsilon(e) := M$.

\item If $e = (e')^T$, then 
\[
\Upsilon(e) := 
\begin{cases}
\rho_{\varphi}(\Upsilon(e')) & \text{if } \alpha \neq 1 \neq \beta; \cr
\rho_{\mathrm{col}_\alpha \to \mathrm{row}_\alpha}(\Upsilon(e')) & \text{if } \alpha \neq 1 = \beta; \cr
\rho_{\mathrm{row}_\beta \to \mathrm{col}_\beta}(\Upsilon(e')) & \text{if } \alpha = 1 \neq \beta; \cr
e' & \text{if } \alpha = 1 = \beta,
\end{cases}
\]
where $\varphi$ maps $\mathrm{col}_\alpha$ to $\mathrm{row}_\alpha$ and $\mathrm{row}_\beta$ to $\mathrm{col}_\beta$.

\item If $e = \mathbf{1}(e')$, then 
\[
\Upsilon(e) := 
\begin{cases}
\mathbf{1}(\pi_{\{\mathrm{row}_\alpha\}}(\Upsilon(e'))) & \text{if } \alpha \neq 1; \cr
\mathbf{1}(\pi_\emptyset(\Upsilon(e'))) & \text{if } \alpha = 1.
\end{cases}
\]

\item If $e = \mathrm{diag}(e')$, then
\[
\Upsilon(e) := 
\begin{cases}
\sigma_{\{\mathrm{row}_\alpha,\mathrm{col}_\alpha\}}(\Upsilon(e') \Join \mathbf{1}(\rho_{\mathrm{row}_\alpha \to \mathrm{col}_\alpha}(\Upsilon(e')))) & \text{if } \alpha \neq 1; \cr
\Upsilon(e') & \text{if } \alpha = 1.
\end{cases}
\]

\item If $e = e_1 \cdot e_2$ where $\mathcal{S}(e_1) = \alpha \times \gamma$ and $\mathcal{S}(e_2) = \gamma \times \beta$, then we consider two cases. If $\gamma = 1$, then $\Upsilon(e) := \Upsilon(e_1) \Join \Upsilon(e_2)$. If $\gamma \neq 1$, then $\Upsilon(e) := \zeta_{C,2}(\rho_{\varphi_1}(\Upsilon(e_1)),\rho_{\varphi_2}(\Upsilon(e_2)))$, where $\varphi_1(\mathrm{col}_\gamma) = \varphi_2(\mathrm{row}_\gamma) = C \notin \{\mathrm{row}_\alpha, \mathrm{col}_\beta\}$ and $\varphi_1$ and $\varphi_2$ are the identity otherwise.

\item If $e = e_1 + e_2$, then $\Upsilon(e) := \Upsilon(e_1) \cup \Upsilon(e_2)$.

\item If $e = e_1 \circ e_2$, then $\Upsilon(e) := \Upsilon(e_1) \Join \Upsilon(e_2)$.
\end{itemize}
It is straightforward to verify by induction on the structure of $e$ that $\Upsilon(e)$ satisfies the given properties.
\end{proof}

\subsection{Simulating \ARACTWO in \ML}

In order to simulate \ARACTWO in \ML, we equip
$\mathbf{att}$ with some linear ordering $<$.

Again we assume that $\mathbf{rel} = \mathbf{matvar}$. Let us fix an injective function $\Psi: \mathbf{att} \to \mathbf{size} \setminus \{1\}$.

Let $X \subseteq \{A_1, A_2\}$ be a relation schema for some $A_1$ and $A_2$ with $A_1 < A_2$. We associate to $X$ an element $\Theta(X) \in \mathbf{size} \times \mathbf{size}$ as follows. We have
\[
\Theta(X) :=
\begin{cases}
\Psi(A_1) \times \Psi(A_2) & \text{if } X = \{A_1,A_2\}; \cr
\Psi(A) \times 1 & \text{if } X = \{A\} \text{ for some } A; \cr
1 \times 1 & \text{if } X = \emptyset.
\end{cases}
\]

Let $\mathcal{S}$ a database schema on a set $N$ of relation
names of arities at most $2$.  We associate to $\mathcal{S}$ a
matrix schema $\Theta(\mathcal{S})$ on $N$
as follows. For $R \in N$, we set
$(\Theta(\mathcal{S}))(R) := \Theta(\mathcal{S}(R))$.

Let $D$ be a domain assignment. We associate to $D$ a size assignment $\sigma(D)$ where, for $A \in \mathbf{att}$, $(\sigma(D))(D(A)) = |D(A)|$. If every element in the range of a domain assignment $D$ is of the form $\{1,\ldots,n\}$ for some $n$, then we say that $D$ is \emph{consecutive}. 

Let $D$ be a consecutive domain assignment. Given a relation $r:
\mathcal{T}_D(X) \to K$ with $X \subseteq \{A_1,A_2\}$ and $A_1 <
A_2$, we associate a matrix $\Mat_D(r)$ conforming to $\Theta(X)$
by $\sigma(D)$ as follows. We define $(\Mat_D(r))_{i,j} := r(t)$,
where $t$ is (1) the tuple with $t(A_1) = i$ and $t(A_2) = j$ if
$|X| = 2$; (2) the tuple with $t(A) = i$ and $j = 1$ if $X =
\{A\}$ for some $A$; and (3) the unique tuple of
$\mathcal{T}_D(X)$ if $X = \emptyset$.

Let $\mathcal{S}$ a database schema on a set $N$ of relation
names of arities at most $2$, and let $\mathcal{I}$ be a database of $\mathcal{S}$
instance with respect to $D$. We associate to $\mathcal{I}$ a
matrix instance $\Mat_D(\mathcal{I})$ of
$\Mat(\mathcal{S})$ with respect to $\sigma(D)$ as follows. For
$R \in N$, we set $(\Mat_D(\mathcal{I}))(R) :=
\Mat_D(\mathcal{I}(R))$.

\begin{example}
Recall $\mathcal{I}$, $\mathcal{S}$, and $D$ from Example~\ref{ex:db_instance}. To reduce clutter, assume that $\mathbf{att} = \mathbf{size} \setminus \{1\}$ and that $\Psi$ is the identity function. Take $\text{student} < \text{dptm}$. We have that $(\Theta(\mathcal{S}))(\text{no\texttt{\_}courses}) = \text{student} \times \text{dptm}$ and $(\Theta(\mathcal{S}))(\text{course\texttt{\_}fee}) = \text{dptm} \times 1$. Consider domain assignment $D'$ and database instance $\mathcal{I}'$ obtained from $D$ and $\mathcal{I}$, respectively, by replacing $\text{Alice}$ by $1$, $\text{Bob}$ by $2$, $\text{CS}$ by $1$, $\text{Math}$ by $2$, and $\text{Bio}$ by $3$. Note that $D'$ is consecutive. The instance $\Mat_{D'}(\mathcal{I}')$ is shown in Figure~\ref{fig:matrixdb_instance}.
\end{example}

We now show that every \ARACTWO expression can be simulated by an \ML expression.
\begin{lemma}\label{lem:ARAtoML}
For each \ARACTWO expression $e$ over a database schema
$\mathcal{S}$ of arity at most $2$, there exists a \ML expression
$\Phi(e)$ over matrix schema $\Theta(\mathcal{S})$ such that (1)
$\Theta(\mathcal{S}(e)) = (\Theta(\mathcal{S}))(\Phi(e))$ and (2)
for all consecutive domain assignments $D$ and database instances
$\mathcal{I}$ with respect to $D$, we have $\Mat_D(e(\mathcal{I})) = (\Phi(e))(\Mat_D(\mathcal{I}))$.
\end{lemma}
\begin{proof}
Assume that $\mathcal{S}(e) \subseteq \{A_1,A_2\}$. We explicitly define $\Phi(e)$.
\begin{itemize}
\item If $e = R$ is a relation name of $\mathcal{S}$, then $\Phi(e) := R$.

\item If $e = e_1 \cup e_2$, then $\Phi(e) := \Phi(e_1) + \Phi(e_2)$.

\item If $e = \hat\pi_A(e')$ and $A_1 < A_2$, then 
\[
\Phi(e) :=
\begin{cases}
\Phi(e') \cdot \mathbf{1}(\Phi(e')^T) & \text{if } A = A_2 \text{ and } \mathcal{S}(e') = \{A_1, A_2\}; \cr
\mathbf{1}(\Phi(e'))^T \cdot \Phi(e') & \text{otherwise}.
\end{cases}
\]

\item If $e = \sigma_{Y}(e')$, then 
\[
\Phi(e) :=
\begin{cases}
\Phi(e') & \text{if } |Y| \leq 1; \cr
\Phi(e') \circ \mathrm{diag}(\mathbf{1}(\Phi(e'))) & \text{if } |Y| = 2.
\end{cases}
\]

\item If $e = \rho_\varphi(e')$ and $A_1 < A_2$, then
\[
\Phi(e) :=
\begin{cases}
\Phi(e')^T & \text{if } \varphi(A_1) > \varphi(A_2) \text{ and } \mathcal{S}(e') = \{A_1, A_2\}; \cr
\Phi(e') & \text{otherwise}.
\end{cases}
\]

\item If $e = \mathbf{1}(e')$, then
\[
\Phi(e) :=
\begin{cases}
\mathbf{1}(\Phi(e')) \cdot \mathbf{1}(\Phi(e'))^T & \text{if } \mathcal{S}(e') = \{A_1, A_2\}; \cr
\mathbf{1}(\Phi(e')) & \text{otherwise}.
\end{cases}
\]

\item If $e = e_1 \Join e_2$ and $A_1 < A_2$, then
\[
\Phi(e) :=
\begin{cases}
\Phi(e_1) \circ \Phi(e_2) & \text{if } \mathcal{S}(e_1) = \mathcal{S}(e_2); \cr
s(e_1,e_2) \circ \Phi(e_2) & \text{if } \mathcal{S}(e_1) = \emptyset; \cr
\Phi(e_1) \circ s(e_2,e_1) & \text{if } \mathcal{S}(e_2) = \emptyset; \cr
\Phi(e_1) \cdot \Phi(e_2)^T & \text{if } \mathcal{S}(e_1) = \{A_1\} \text{ and } \mathcal{S}(e_2) = \{A_2\}; \cr
(\Phi(e_1) \cdot \Phi(e_2)^T)^T & \text{if } \mathcal{S}(e_1) = \{A_2\} \text{ and } \mathcal{S}(e_2) = \{A_1\}; \cr
\mathrm{diag}(\Phi(e_1)) \cdot \Phi(e_2) & \text{if } \mathcal{S}(e_1) = \{A_1\} \text{ and } \mathcal{S}(e_2) = \{A_1,A_2\}; \cr
(\Phi(e_1)^T \cdot \mathrm{diag}(\Phi(e_2)))^T & \text{if } \mathcal{S}(e_1) = \{A_1,A_2\} \text{ and } \mathcal{S}(e_2) = \{A_1\}; \cr
\Phi(e_1) \cdot \mathrm{diag}(\Phi(e_2)) & \text{if } \mathcal{S}(e_1) = \{A_1,A_2\} \text{ and } \mathcal{S}(e_2) = \{A_2\}; \cr
(\mathrm{diag}(\Phi(e_1)) \cdot \Phi(e_2)^T)^T & \text{if } \mathcal{S}(e_1) = \{A_2\} \text{ and } \mathcal{S}(e_2) = \{A_1,A_2\},
\end{cases}
\]
where $s(e,e')$ denotes $\mathbf{1}(\Phi(e')) \cdot \Phi(e) \cdot \mathbf{1}(\Phi(e')^T)^T$.

\item If $e = \zeta_{A_3,2}(e_1,e_2)$ with $\mathcal{S}(e_1) = \{A_1,A_3\}$ and $\mathcal{S}(e_2) = \{A_2,A_3\}$, then
\[
\Phi(e) :=
\left(\Phi(e_1)^{T(A_1,A_3)} \cdot \Phi(e_2)^{T(A_3,A_2)}\right)^{T(A_1,A_2)},
\]
where $(\cdot)^{T(A,B)}$, for attributes $A$ and $B$, denotes identity if $A < B$ and transpose if $A > B$.
\end{itemize}

The last bullet covers the case where $|\mathcal{S}(e_1) \sdiff \mathcal{S}(e_2)| = 2$, where $\sdiff$ denotes symmetric difference. If $|\mathcal{S}(e_1) \sdiff \mathcal{S}(e_2)| \leq 1$, then $\zeta_{A_3,2}(e_1,e_2) \equiv \hat\pi_{A_3}(e_1 \Join e_2)$ is expressible in $\ARAm(2)$ (since then $|\mathcal{S}(e_1 \Join e_2)| \leq 2$) and so $\Phi$ can be extended to cover this case as well.

It is straightforward to verify by induction on the structure of $e$ that $\Phi(e)$ satisfies the given properties.
\end{proof}

We remark that the number of cases in the expression for $\Phi(e)$ with $e = e_1 \Join e_2$ in the above proof can be significantly reduced if we assume that $K$ is commutative (i.e., join is commutative).

\subsection{Relationship with $\ARAm(3)$ and complexity}

Corollary~\ref{cor:ARA_kp1+k}, Lemma~\ref{lem:MLtoARA}, and Lemma~\ref{lem:ARAtoML} together establish the equivalence of \ML with the language $\ARAm(3)$ restricted to database schemas and output relations of arity at most $2$.

\begin{theorem}\label{thm:MLequivARA}
For each $\ARAm(3)$ expression $e$ over a database schema
$\mathcal{S}$ of arity at most $2$ and with $|\mathcal{S}(e)|
\leq 2$, there exists a \ML expression $e'$ such that
$\Mat_D(e(\mathcal{I})) = e'(\Mat_D(\mathcal{I}))$ for all
consecutive domain assignments $D$ and instances $\mathcal{I}$
with respect to $\mathcal{S}$ over $D$.

Conversely, for each \ML expression $e$ over a matrix schema $\mathcal{S}$, there exists an $\ARAm(3)$ expression $e'$ such that $\Rel_{\mathcal{S}(e),\sigma}(e(\mathcal{I})) = e'(\Rel_{\mathcal{S},\sigma}(\mathcal{I}))$ for all size assignments $\sigma$ and matrix instances $\mathcal{I}$ of $\mathcal{S}$ with respect to $\sigma$.
\end{theorem}

As to complexity, we note that the translations $\Upsilon$ and
$\Phi$ given by Lemmas \ref{lem:MLtoARA} and \ref{lem:ARAtoML},
taken at face value, are exponential.  They can, however, be
readily adapted to become linear (for fixed schemas; quadratic
when the schema is part of the input).  The adaptations that need
to be done are as follows.

For a \ML expression $e'$ with $\mathcal{S}(e') = \alpha \times
1$ with $\alpha \neq 1$, there is a constant-length expression
$\mathrm{Tp}_\alpha$ with $\mathcal{S}(\mathrm{Tp}_\alpha) =
\alpha \times 1$. Indeed, since $\alpha$ is a size term of
$\mathcal{S}(e')$ distinct from $1$, there is a matrix variable
$M$ with $\mathcal{S}(M)$ equal to either $\alpha \times \gamma$
or $\gamma \times \alpha$ for some $\gamma$. Taking
$\mathrm{Tp}_\alpha := \mathbf{1}(M)$ in the former case and
$\mathrm{Tp}_\alpha := \mathbf{1}(M^T)$ in the latter case, we
have $\mathcal{S}(\mathrm{Tp}_\alpha) = \alpha \times 1$ as
desired. The only source of exponential growth in
Lemma~\ref{lem:MLtoARA} is the expression
$\sigma_{\{\mathrm{row}_\alpha,\mathrm{col}_\alpha\}}(\Upsilon(e')
\Join \mathbf{1}(\rho_{\mathrm{row}_\alpha \to
\mathrm{col}_\alpha}(\Upsilon(e'))))$ appearing in the
$\mathrm{diag}(e')$ case, which is equivalent to
$\sigma_{\{\mathrm{row}_\alpha,\mathrm{col}_\alpha\}}(\Upsilon(e')
\Join \rho_{\mathrm{row}_\alpha \to
\mathrm{col}_\alpha}(\Upsilon(\mathrm{Tp}_\alpha)))$.

For the converse translation,
we observe that, for an \ARA expression $e'$ with $X := \mathcal{S}(e')
\subseteq \{A_1,A_2\}$, there is a constant-length expression
$\mathrm{Tp}_X$ with $\mathcal{S}(\mathrm{Tp}_X) = X$. Indeed, if
$A \in \mathcal{S}(e')$, then there exists $A' \in
\mathcal{S}(R_{A'})$ for some relation name $R_{A'}$ such that
$A'$ is compatible with $A$. Taking $\mathrm{Tp}_X := \Join_{A
\in X} \rho_{A' \to A}(\pi_{\{A'\}}(R_{A'}))$ if $X \neq
\emptyset$ and $\mathrm{Tp}_\emptyset := \pi_{\emptyset}(R)$ for
some relation name $R$, we have $\mathcal{S}(\mathrm{Tp}_X) = X$
as desired. Replacing each occurrence of $\mathbf{1}(\Phi(e'))$
by the equivalent expression
$\mathbf{1}(\Phi(\mathrm{Tp}_{\mathcal{S}(e')}))$ and each
occurrence of $\mathbf{1}(\Phi(e')^T)$ by the equivalent
expression $\mathbf{1}(\Phi(\mathrm{Tp}_{\mathcal{S}(e')})^T)$ in
the proof of Lemma~\ref{lem:ARAtoML} avoids
exponential growth.

\subsection{Indistinguishability}

Using a
recent result by Geerts on indistinguishability in \ML
\cite{floris_lagraphs}, we can
also relate $\ARAm(3)$ to $\mathsf{C}^3$, the three-variable
fragment of first-order logic with counting \cite{otto_bounded}.
Let $A_1$ and $A_2$
be matrices of the same dimensions $m \times n$.
We view $A_1$ and $A_2$ as instances of a schema $\mathcal S$
on a single matrix name $M$ with $\mathcal S(M)=\alpha \times
\beta$, with respect to the size assignment $\sigma$ that maps $\alpha$ to
$m$ and $\beta$ to $n$.
We say that $A_1$ and $A_2$ are indistinguishable in
\ML, denoted by $A_1 \equiv_{\MLm} A_2$, if for each \ML
expression $e$ over $\mathcal S$ with $\mathcal{S}(e) =
\emptyset$,
we have $e(A_1) = e(A_2)$.
Similarly, one can define indistinguishability of
binary $K$-relations $r_1$ and $r_2$ in $\ARAm(3)$, denoted by $r_1
\equiv_{\ARAm(3)} r_2$. This leads to the following corollary to
Theorem~\ref{thm:MLequivARA}. Let $s=\alpha\times\beta$.

\begin{corollary}\label{cor:MLequivARA}
$A_1 \equiv_{\MLm} A_2$
if and only if $\Rel_{s,\sigma}(A_1) \equiv_{\ARAm(3)} \Rel_{s,\sigma}(A_2)$.
\end{corollary}

Geerts's result concerns finite undirected graphs $G_1$ and $G_2$
with the same number of nodes.
Recall that $G_1$ and $G_2$ are called indistinguishable in
$\mathsf{C}^3$, denoted by $G_1 \equiv_{\mathsf{C}^3} G_2$, if
each $\mathsf{C}^3$-sentence over a single binary
relation variable has the same truth value on $G_1$ and $G_2$.
Denote the adjacency matrix of $G$ by $\mathsf{Adj}(G)$.

\begin{theorem}[\cite{floris_lagraphs}]\label{thm:geerts_indist}
Fix $K$ to be the field of complex numbers.
Then $\mathsf{Adj}(G_1) \equiv_{\MLm} \mathsf{Adj}(G_2)$ if and only if $G_1 \equiv_{\mathsf{C}^3} G_2$.
\end{theorem}

We can immediately conclude the following, still fixing $K$ to be
the field of complex numbers.
\begin{corollary}\label{cor:geerts_indist}
$G_1 \equiv_{\mathsf{C}^3} G_2$ if and only if
$\Rel_{s,\sigma}(\mathsf{Adj}(G_1)) \equiv_{\ARAm(3)}
\Rel_{s,\sigma}(\mathsf{Adj}(G_2))$ (for suitable $s$ and
$\sigma$).
\end{corollary}

\section{Conclusion} \label{seconcl}

In related work, Yan, Tannen, and Ives consider provenance
for linear algebra operators \cite{provsemirings_la}.  In that
approach, provenance tokens represent not the matrix entries (as in
our work), but the matrices themselves.  Polynomial expressions
(with matrix addition and matrix multiplication) are derived to
show the provenance of linear algebra operations applied to
these matrices.  

\newcommand{\laggr}{\mathcal{L}_{\rm Aggr}}
Our result that every matrix query expressible in $\ARAm(3)$ is also
expressible in \ML provides a partial converse to the
observation already made in the original paper
\cite{matlang_icdt}, to the effect that \ML can be expressed
in $\laggr(3)$: the relational calculus with summation
and numerical functions \cite{libkin_sql}, restricted to three base
variables.\footnote{$\laggr$ is a two-sorted logic with
base variables and numerical variables.} This observation was
made in the extended setting of \ML that allows arbitrary
pointwise functions (Remark~\ref{rempoint}).
For the language considered here, $\ARAm(3)$
provides a more appropriate upper bound for comparison, and
$\ARAm(3)$ is still a natural fragment of $\laggr(3)$.

When allowing arbitrary pointwise functions in \ML, we actually
move beyond the positive relational algebra, as queries involving
negation can be expressed.  For example, applying the function $x
\land \lnot y$ pointwise to the entries of two $n\times n$
boolean matrices representing two binary relations $R$ and $S$ on
$\{1,\dots,n\}$, we obtain the set difference $R - S$.  It is an
interesting research question to explore expressibility of
queries in \ML in this setting.  For example, consider the
following $\laggr(3)$ query on two matrices $M$ and $N$: \[
\forall i \exists j \forall k \forall x (M(i,k,x) \to \exists i
\, N(j,i,x)) \]
Here, $M(i,k,x)$ means that $M_{i,k}=x$, and similarly for
$N(j,i,x)$.

The above query, which does not even use summation,
reuses the base variable $i$ and checks whether each row of $M$,
viewed as a set of entries, is included in some row of $N$, again
viewed as a set of entries.  We conjecture that the query is not
expressible in \ML with arbitrary pointwise functions.
Developing techniques for showing this is an interesting
direction for further research.

Finally, recall that our main result
Corollary~\ref{cor:ARA_kp1+k} assumes that $K$ is commutative.
It should be investigated whether or not this result still holds in the
noncommutative case.

\subsection*{Acknowledgements}
We thank Floris Geerts for inspiring discussions. R.B.\ is a postdoctoral fellow of the Research Foundation -- Flanders (FWO).

\bibliographystyle{plainurl}
\bibliography{database}

\end{document}